\documentclass[12pt]{article}
\usepackage{amssymb,latexsym}
\usepackage{amsfonts}
\usepackage{amsmath}
\usepackage[colorlinks,linkcolor=blue,anchorcolor=blue,citecolor=blue]{hyperref}
\usepackage{algorithm}
\usepackage{enumerate}
\usepackage{algpseudocode}
\usepackage{verbatim}
\usepackage{graphicx}
\usepackage{amsthm}

\def\F{\mathbb{F}}
\def\E{\mathbb{F}_{q_1}}

\def\Tr{\text{\rm Tr}}
\def\wt{\text{\rm wt}}
\def\Ker{\text{\rm Ker}}

\def\Im{\text{\rm Im}}

\def\D{\textrm{D}}

\def\Prj{\textrm{Prj}}

\newtheorem{theorem}{Theorem}[section]
\newtheorem{lemma}[theorem]{Lemma}
\newtheorem{example}[theorem]{Example}
\newtheorem{corollary}[theorem]{Corollary}
\newtheorem{proposition}[theorem]{Proposition}

\newtheorem{remark}[theorem]{Remark}

\numberwithin{equation}{section}

{\centering
\title{Some three-weight linear codes and their complete weight enumerators and weight hierarchies }
}
\author{Xiumei Li$^{1}$, Zongxi Chen$^{1}$, Fei Li$^{2}$\footnote{Corresponding author.\protect\\
        E-mail address: lxiumei2013@qfnu.edu.cn(X. Li), cczxlf@163.com(F. Li)}\\
    \hspace{0.5cm}\small $^{1}$School of Mathematical Sciences, Qufu Normal University, Qufu  273165, China\\
    \small $^{2}$Faculty of School of Statistics and Applied Mathematics, \\ \small  Anhui University of Finance and Economics, Bengbu {\rm 233030}, Anhui, China}

\date{}

\begin{document}
\maketitle

\begin{abstract}
Linear codes with a few weights can be applied to secrete sharing, authentication codes, association schemes and strongly regular graphs. For an odd prime power $q$, we construct a class of three-weight $\F_q$-linear codes from quadratic functions via a bivariate construction and then determine the complete weight enumerators and weight hierarchies of these linear codes completely. This paper generalizes 
some results in Li et al. (2022) and Hu et al. (2024). 
\end{abstract}



\section{Introduction}

\label{intro}


Let $q=p^m$ for an odd prime number $p$ and $ \mathbb{F}_{q^s} $ be the finite field with $ q^{s} $
elements. Denote by $\mathbb{F}_{q^s}^{*}$ the set of the nonzero elements of $ \mathbb{F}_{q^s}$ and $\Tr_q^{q^s}$ the trace function from $\mathbb{F}_{q^s}$ onto $\mathbb{F}_q$, respectively.

A $k$-dimensional subspace $C$ of $ \mathbb{F}_{q}^{n} $ over $\F_q$ is called an $[n,k,d]_q$ linear code of length $n$ with  minimum (Hamming) distance $d$, where $d$ determines the error correction capability of $C$. For $i\in\{1,2,\cdots,n\}$, let $A_{i}$ be the number of codewords in $C$ with Hamming weight $i$ and the sequence $(1,A_{1},\cdots,A_{n})$ be the weight distribution of $C$.
If there are $t$ nonzero elements in the sequence $(A_{1},\cdots,A_{n})$, then the code $C$ is called a $t$-weight code.
Let $\F_q=\{\omega_0=0,\omega_1,\cdots,\omega_{q-1}\}$. For a vector $c=(c_0,c_1,\cdots,c_{n-1})\in\F_q^n$, the complete weight enumerator $\omega[c]$ of $c$ is defined by 
$$\omega[c]=\omega_0^{k_0}\omega_1^{k_1}\cdots\omega_{q-1}^{k_{q-1}},$$
where $\sum\limits_{j=0}^{q-1}k_j=n,k_j$ is the number of components of $c$ that equals to $\omega_j$. The complete weight enumerator of $C$ is defined by 
$$\omega[C]=\sum\limits_{c\in C}\omega[c].$$
It's easy to see that the complete weight enumerator of a code provides more detailed information of the weight structure than the ordinary weight distribution.
The complete weight enumerator of a code $C$
could be used to calculate deception probabilities of certain authentication codes and study the Walsh transform of monomial functions over finite fields. Linear codes with a few weights have important applications in authentication codes \cite{DH07},
association schemes \cite{CG84}, secret sharing \cite{YD06} and strongly regular graphs \cite{CK86}. Thus, the determination of the complete weight enumerators of specific linear codes is a significant research topic in coding theory.

 For an $[n,k,d]_q$ linear code $C$, 
let
$ [C,r]_{q} $ be the set of all the $r$-dimensional subspaces of $C$, where $ 1\leq r\leq k$.
For a subspace $ H \in [C,r]_{q}$, the support of $H$ is defined by
$$ \textrm{Supp}(H)=\Big\{i:1\leq i\leq n, c_i\neq 0 \ \ \textrm{for some $c=(c_{1}, c_{2}, \cdots , c_{n})\in H$}\Big\}.$$
The $r$-th generalized Hamming weight of $C$ is defined by
$$
d_{r}(C)=\min\Big\{|\textrm{Supp}(H)|:H\in [C,r]_{q}\Big\}, \ 1\leq r\leq k.
$$
The sequence $ \{d_{1}(C),d_{2}(C),\cdots,d_{k}(C)\}$ is called the weight hierarchy of $C$. 
The weight hierarchies can be used to deal with $t$ - resilient functions, and trellis or branch complexity of linear codes \cite{CGHFRS85, TV95}, computation of state and branch complexity profiles \cite{KTFL93,F94} and so on.
 The weight hierarchies for some well-known classes of codes were determined, such as Hamming codes, Reed-Muller codes, Reed-Solomon codes,Golay codes and some cyclic codes \cite{BLV14,B19,CC97,HP98,JL97,WJ91,XL16,YL15}. However, determining the weight hierarchy of linear code is relatively challenging. During the past three decades, there were research results
about weight hierarchies of some classes of linear codes \cite{JF17,LZW23,LF18,LF21,LL20-0,LL21,LL22,LW19,WZ94,WW97}.

Let $ D= \{d_{1},d_{2},\cdots,d_{n}\}$ be a subset of $\mathbb{F}_{q^s}^{\ast}$, Ding et al. \cite{DN07} proposed a generic construction of linear codes as follows:
\begin{eqnarray}\label{defcode0}
         \small{C_{D}=\{\left( \Tr_q^{q^s}(xd_1), \Tr_q^{q^s}(xd_2),\ldots, \Tr_q^{q^s}(xd_{n})\right):x\in \mathbb{F}_{q^{s}}\},}
\end{eqnarray}
where $D$ is called the defining set of $C_{D}$.
In recent years, Ding and Niederreiter's construction was extended to the bivariate form, namely, 
\begin{eqnarray}\label{defcode1}
         C_{D}=\{( \Tr_{q}^{q^{s_1}}(ux)+\Tr_{q}^{q^{s_2}}(vy))_{(x,y)\in D}:(u,v)\in \F_{q^{s_1}}\times\F_{q^{s_2}}\},
\end{eqnarray}
where $s_i(i=1,2)$ are positive integers and $D$ is a subset of $(\F_{q^{s_1}}\times\F_{q^{s_2}})\backslash\{(0,0)\}$. Furthermore, 
there have been extensive researches on the parameters and the weight distributions of linear codes constructed using the above two construction methods,  leading to fruitful results (see \cite{DJ16,DD14,DD15,DLN08,DN07,JLF19,LYF18,LL21,LL22,LL24,S22,TXF17} and the references therein).

 Let $f(x)$ be a quadratic form function from $\F_{q^{s_1}}$ to $\F_q$ and $\alpha\in\F_q^*,\beta\in\F_q$, we define $\overline{D}_\beta$ and $D$ as follows:
\begin{align}\label{set:D1}
 &\overline{D}_{\beta}=\Big\{x\in \F_{q^{s_1}}|f(x)=\beta\Big\},\\
&D=\Big\{(x,y)\in (\F_{q^{s_1}}\times\F_{q^{s_2}})\backslash\{(0,0)\}: f(x)+\Tr_{q}^{q^{s_2}}(\alpha y)=\beta\Big\}. 
\end{align}
In the case $\beta=0$, the complete weight enumerator of $C_{\overline{D}_{\beta}\setminus \{(0,0)\}}$ defined in \eqref{defcode0} was determined by Zhang and Fan et al. \cite{ZFPT17},
while the weight hierarchies of $C_{\overline{D}_{\beta}\setminus \{(0,0)\}}$ was completely determined by Wan \cite{WZ94}, Wan and Wu \cite{WW97}, which were deduced by application of finite projective geometry. When $\beta \neq 0$, the complete weight enumerator of $C_{\overline{D}_{\beta}}$ in \eqref{defcode0} was determined by Du and Wan \cite{DW17}, while
the weight hierarchy of $\F_p$-linear codes $C_{\overline{D}_{\beta}}$ was completely determined by Li \cite{LF21}, Li and Li \cite{LL20-0} firstly using a different method and the weight hierarchy of $\F_{2^s}$-linear codes $C_{\overline{D}_{\beta}}$ was completely determined by Liu and Zheng et al. \cite{LZW23}. Notice that the results in \cite{LF21,LL20-0} can be directly extended to the case of $\F_q$-linear code $C_{\overline{D}_{\beta}}$. 

In the case $\alpha \neq 0, \beta=0$, the weight distribution and the weight hierarchy of $\F_p$-linear code $C_{D}$ defined in \eqref{defcode1} was determined by Li and Li \cite{LL22} and Hu et al. \cite{HLL24}. In this paper, for the case of clarity, we consider the case $\alpha(\neq 0), \beta\in \F_q$ and study the complete weight enumerators and the weight hierarchies of $\F_q$-linear codes $C_{D}$ in \eqref{defcode1}.

In Section 2, we set the main notations and give some
properties of $q$-ary quadratic forms. In Section 3, we determine the parameters of the present linear codes and examine their explicit weight enumerators. In Section 4, we determine their weight hierarchies completely. Finally, Section 5 concludes the paper.

\section{preliminaries}

\subsection{Some notations fixed throughout this paper}

  For convenience, we fix the following notations. 
\begin{itemize}
\item $p$ is an odd prime number and $q=p^m$.
\item $s_1,s_2$ are positive integers and $s=s_1+s_2$.
\item $q_i= q^{s_i}, i=1,2, \mathbb{F}=\mathbb{F}_{q_1}\times\mathbb{F}_{q_2},\mathbb{F}^\star=\mathbb{F}\Big\backslash\{(0,0)\}$.
\item $\zeta_{p}=\exp(\frac{2\pi i}{p})$ is a primitive $p$-th root of unity. 
\item $p^*=(-1)^{\frac{p-1}{2}}p$. 
\item $\eta$ is the quadratic character of $\mathbb{F}_{q}^{\ast}$ with letting $\eta(0)=0$.
\item $\upsilon(\cdot)$ is defined by the integer-values function of $\mathbb{F}_{q}$, $\upsilon(x)=q-1$ if $x=0$, otherwise $\upsilon(x)=-1$.
\item $\Big\langle\alpha_{1},\alpha_{2},\cdots,\alpha_{r}\Big\rangle$ is the $\mathbb{F}_q$-linear space spanned by elements $\alpha_{1},\alpha_{2},\cdots,\alpha_{r}$ in this paper.
\end{itemize}

\begin{lemma}[{\cite[Theorem 5.15 and 5.33]{LN97}}] With the symbols and notations above and let $f(x)=a_2x^2+a_1x+a_0\in\mathbb{F}_q[x]$, where $a_2\neq 0,q=p^m$. Then \\
\rm{(1)} $\sum\limits_{c\in\mathbb{F}_q^*}\eta(c)\zeta_p^{\mathrm{Tr}_p^q(c)}=(-1)^{m-1}\eta(-1)p^m(p^*)^{-\frac{m}{2}}$.\\
\rm{(2)} $\sum\limits_{c\in\mathbb{F}_q}\zeta_p^{\mathrm{Tr}_p^q(f(c))}=\zeta_p^{\mathrm{Tr}_p^q(a_0-a_1^2(4a_2)^{-1})}\eta(-a_2)(-1)^{m-1}p^m(p^*)^{-\frac{m}{2}}$.
\end{lemma}

\begin{lemma}\label{cor:1} For $b\in\F_{q}$, we have
\[\sum\limits_{z\in \mathbb{F}_{q}^{\ast}} \eta(z)^k\zeta_p^{\mathrm{Tr}_p^q(zb)}=\left\{\begin{array}{ll}
    		v(b), & k \textrm{\ is even}, \\
    	\eta(b)(-1)^{m-1}\eta(-1)p^m(p^*)^{-\frac{m}{2}},  & k \textrm{\ is odd}.
    	\end{array}
    	\right.\]
\end{lemma}

\subsection{Quadratic form}

A quadratic form $f$ over $\E$ is a map $\E\to \F_q$ such that $f(ax)=a^2f(x)$ for all $a\in\F_q$ and \begin{equation*}\label{eq:F}
    F(x,y)=\frac{1}{2}\Big(f(x+y)-f(x)-f(y)\Big), \textrm{for any}\  x,y\in\E,
\end{equation*} is a symmetric bilinear form. 
Identifying $\mathbb{F}_{q_1}$ with an $s_1$-dimensional vector space $\F_q^{s_1}$ over $\F_q$, that is, taking a basis $\upsilon_{1},\upsilon_{2},\cdots,\upsilon_{s_1}$ of $\E$,
there is an $\F_q$-linear isomorphism $\E\simeq\F_q^{s_1}$ defined as
$$x=x_{1}\upsilon_{1}+x_{2}\upsilon_{2}+\cdots+x_{s_1}\upsilon_{s_1} \mapsto \overline{x}=(x_1,x_2,\cdots,x_{s_1}),$$
where $\overline{x}\in \F_q^{s_1}$ is called the coordinate vector of $x$ under the basis $v_1,v_2,\cdots,v_{s_1}$ of $\E$,
then $f$ can be represented by
\begin{align}\label{eq:f}
f(x)=f(\overline{x})=f(x_{1},x_{2},\cdots,x_{s_1})&=\sum_{1\leq i,j\leq s_1}F(v_i,v_j)x_{i}x_{j}
=\overline{x}A\overline{x}^T,
\end{align}
where $A=(F(v_i,v_j))_{s_1\times s_1}, F(v_i,v_j)\in \mathbb{F}_{q}, F(v_i,v_j)=F(v_j,v_i)$ and $\overline{x}^T$ is the transposition of $\overline{x}$. The rank $R$ of $f$ is defined as the codimension of $\F_q$-vector space $$\E^{\perp_f}=\{x\in\E|F(x,y)=0, \textrm{for all } y\in\E\}.$$ Then $|\E^{\perp_f}|=q^{s_1-R}$. We call $f$ non-degenerate if $R=s_1$ and degenerate, otherwise. It's known that there exists an invertible matrix $M$ over $\F_q$ such that
$$
MAM^T=\textrm{diag}(\lambda_{1},\lambda_{2},\cdots,\lambda_{R},0,\cdots,0)
$$
is a diagonal matrix, where $\lambda_{1},\lambda_{2},\cdots,\lambda_{R}\in\F_q^*$.
Let $\Delta_{f}=\lambda_{1}\lambda_{2}\cdots\lambda_{R}$ if $R\neq0$ and $\Delta_{f}=1$, otherwise. 
The sign $\varepsilon_{f}$ of $f$ is defined as $\eta(\Delta_{f})$.

For a subspace $H$ of $\E$, 
if $f$ is restricted to $H$,
it becomes a quadratic form denoted by $f|_{H}$ over $H$ in $r$ variables.
Let $R_{H}$ and $\varepsilon_{H}$ be the rank and the sign of $f|_{H}$ over $H$, respectively.

\begin{lemma}[{\cite[Proposition 1]{LF21}}]\label{lem:1}
Let $f$ be a quadratic form over $\E$ and $H$ be an $r$-dimensional ($r>0$) subspace of $\E$. Then for $c\in \mathbb{F}_{q}$ the number of solutions of the equation $f(x)=c$ in $H$ is 
$$
\Big|H\cap \overline{D}_{c}\Big|=\left\{\begin{array}{ll}
q^{r-1}+\upsilon(c)\eta((-1)^{\frac{R_H}{2}})\varepsilon_{H}q^{r-\frac{R_H+2}{2}},  &\textrm{if $R_H$ is even}, \\
q^{r-1}+\eta((-1)^{\frac{R_H-1}{2}}c)\varepsilon_{H}q^{r-\frac{R_H+1}{2}},  &\textrm{if $R_H$ is odd}.
\end{array}
\right.
$$
\end{lemma}

Let $f$ be a quadratic form over $\E$. For any $x,y \in\E$, there exists a linearized polynomial $L_f$ over $\E$ such that $f(x)=\Tr_q^{q_1}(xL_f(x))$ and
\begin{equation}\label{eq:F} F(x,y)=\Tr_q^{q_1}\Big(xL_{f}(y)\Big)=\Tr_q^{q_1}\Big(yL_{f}(x)\Big). \end{equation}
Let $\mathrm{Im}(L_{f})=\Big\{L_{f}(x):x\in\E\Big\},\ \Ker(L_{f})=\Big\{x\in\E:L_{f}(x)=0\Big\}$ denote the image and the kernel of $L_{f}$, respectively.
If $b\in \mathrm{Im}(L_{f})$, we denote $x_{b}\in \E$ with $L_{f}(x_{b})=-\frac{b}{2}$. For more details, one can refer to \cite[Sec.IV]{HK06} and \cite[Sec.2]{CM23}.

From eq. \eqref{eq:F}, we have
\begin{equation*}
\Ker(L_{f})=\{x\in \E:f(x+y)=f(x)+f(y), \textrm{for any}\  y\in\E\}=\mathbb{F}_q^{\perp_f}
\end{equation*}
and $\textrm {rank} (L_f) = R$.

By a similar proof of Lemma 5 in \cite{TXF17}, we have. 

\begin{lemma}[{\cite[Lemma 5]{TXF17}}]\label{lem:6}
 Let the symbols and notation be as above. Let $f$ be defined in \eqref{eq:f} and $H$ be an $r$-dimensional subspace of \  $ \E$. Then for $b\in \E$, we have the following.\\
 \rm{(1)} $\sum\limits_{x\in H}\zeta_p^{\mathrm{Tr}_p^q(f(x))}=\left\{
\begin{array}{ll}
			\varepsilon_H q^{r}(p^*)^{-\frac{mR_H}{2}}, & 2|R_H,\\
			(-1)^{m-1}\eta(-1)\varepsilon_H q^{r}(p^*)^{-\frac{mR_H}{2}}, & 2\not|R_H.
		\end{array} \right.$\\
 \rm{(2)} For any $z\in \F_q^*$,\\
 \rm{(2.1)} if $b\notin \mathrm{Im}(L_f)$, then we have
$$\sum\limits_{x\in \mathbb{F}_{q_1}}\zeta_p^{\mathrm{Tr}_p^q(zf(x)-\mathrm{Tr}_q^{q_1}(zbx))}=0.$$
 \rm{(2.2)} if $b\in \mathrm{Im}(L_f)$, then we have
$$\sum\limits_{x\in \mathbb{F}_{q_1}}\zeta_p^{\mathrm{Tr}_p^q(zf(x)-\mathrm{Tr}_q^{q_1}(zbx))}=\left\{
\begin{array}{ll}
		\varepsilon_fq^{s_1}(p^*)^{-\frac{mR}{2}}\zeta_p^{-\mathrm{Tr}_p^q(zf(x_b))}, & 2|R,\\
		(-1)^{m-1}\eta(-1)\varepsilon_fq^{s_1}(p^*)^{-\frac{mR}{2}} \eta(z)\zeta_p^{-\mathrm{Tr}_p^q(zf(x_b))}, &2\not|R,
		\end{array} \right.$$
where $x_{b}$ satisfies $L_{f}(x_{b})=-\frac{b}{2}$.
\end{lemma}

\section{The complete weight enumerators}

In this section, we first calculate the length of $C_{D}$ defined in \eqref{defcode1} and the Hamming weight of non-zero codewords of $ C_{D}$.

\begin{lemma}\label{lem:length}
Let $D$ be defined as above and $C_{D}$ be defined in \eqref{defcode1}, define $n=|D|$. Then $$n =q^{s-1}-e,$$
where $e=1$,if $\beta=0$;$e=0$,if $\beta\neq0$. 
\end{lemma}
\begin{proof}
By Lemma \ref{lem:1}, the number of solutions of the equation $f(x)+\Tr_q^{q_2}(\alpha y)=\beta$ in $\E$ is 
\begin{align*}
N(f(x)+\Tr_q^{q_2}(\alpha y)=\beta)
&=\sum_{c_1\in \mathbb{F}_q} N(f(x)=c_1)N(\Tr_q^{q_2}(\alpha y)=\beta-c_1) \\
&=\frac{q_2}{q}\sum_{c_1\in \mathbb{F}_q} N(f(x)=c_1)  \\
&=q^{s-1}.
\end{align*}
Thus, the desired conclusion can be obtained.
\end{proof}

\begin{lemma}\label{lem:N(u,v,t)}
For any non-zero element $(u,v)\in \F$ and $t\in\F_q$,
put
$$N(u,v,t)=\Big\{(x,y)\in \F: f(x)+\Tr_q^{q_2}(\alpha y)=\beta, \Tr_q^{q_1}(ux)+\Tr_q^{q_2}(vy) = t\Big\}.$$ We have the following.
\begin{enumerate}
\item [\rm{(1)}] When $v\in \mathbb{F}_{q_2}\setminus \mathbb{F}_{q}^*\alpha$, we have $N(u,v,t)=q^{s-2}$.
\item [\rm{(2)}] When $v\in \mathbb{F}_{q}^{\ast}\alpha$, we have the following two cases.
 \item [\rm{(2.1)}] If $u\notin \textrm{Im}(L_f)$, then $N(u,v,t)=q^{s-2}$.
  \item [\rm{(2.2)}] If $u\in \textrm{Im}(L_f)$, then
$$
N(u,v,t)
=\left\{\begin{array}{ll}
    		q^{s-2}\Big(1+\varepsilon_f(p^*)^{-\frac{mR}{2}}\upsilon(\beta-zt+z^2f(x_u))\Big), & 2|R, \\
    	q^{s-2}\Big(1+\varepsilon_f(p^*)^{-\frac{m(R-1)}{2}}\eta(\beta-zt+z^2f(x_u))\Big),  & 2\not|R.
    	\end{array}
    	\right.
$$

\end{enumerate}
\end{lemma}
\begin{proof} 
By the orthogonal property of additive characters, we have
\begin{align}\label{eq:N(u,v)}
&\substack{|N(u,v,t)|=\frac{1}{q^{2}}\sum\limits_{(x,y)\in \mathbb{F}}\Big(\sum\limits_{z_{1}\in \mathbb{F}_{q}}\zeta_{p}^{\Tr_p^q(z_{1}(f(x)+\Tr_q^{q_2}(\alpha y)-\beta))}\sum\limits_{z_{2}\in \mathbb{F}_{q}}\zeta_{p}^{\Tr_p^q(z_{2}(\Tr_q^{q_1}(ux)+\Tr_q^{q_2}(vy)-t))}\Big)}  \\ \nonumber
&\substack{=\frac{1}{q^{2}}\sum\limits_{(x,y)\in \mathbb{F}}\Big(\big(1+\sum\limits_{z_{1}\in \mathbb{F}_{q}^*}\zeta_{p}^{\Tr_p^q(z_{1}(f(x)+\Tr_q^{q_2}(\alpha y)-\beta))}\big)\big(1+\sum\limits_{z_{2}\in \mathbb{F}_{q}^*}\zeta_{p}^{\Tr_p^q(z_{2}(\Tr_q^{q_1}(ux)+\Tr_q^{q_2}(vy)-t))}\big)\Big)}  \\ \nonumber
&\substack{=q^{s-2}+\frac{1}{q^{2}}\Big(\sum\limits_{z_{1}\in \mathbb{F}_{q}^{\ast}}\sum\limits_{(x,y)\in \mathbb{F}}\zeta_{p}^{\Tr_p^q(z_{1}(f(x)+\Tr_q^{q_2}(\alpha y)-\beta))}+ \sum\limits_{z_{2}\in \mathbb{F}_{q}^{\ast}}\sum\limits_{(x,y)\in \mathbb{F}}\zeta_{p}^{\Tr_p^q(z_{2}(\Tr_q^{q_1}(ux)+\Tr_q^{q_2}(vy)-t))}\Big)} \\ \nonumber
&\substack{+\frac{1}{q^{2}}\sum\limits_{z_{1}\in \mathbb{F}_{q}^{\ast}}\sum\limits_{z_{2}\in \mathbb{F}_{q}^{\ast}}\sum\limits_{(x,y)\in \mathbb{F}}\zeta_{p}^{\Tr_p^q(z_{1}(f(x)+\Tr_q^{q_2}(\alpha y)-\beta)+z_{2}(\Tr_q^{q_1}(ux)+\Tr_q^{q_2}(vy)-t))}}  \\  \nonumber
&\substack{=q^{s-2}+q^{-2}\sum\limits_{z_{1}\in \mathbb{F}_{q}^{\ast}}\zeta_p^{-\Tr_p^q(z_1\beta)}\sum\limits_{z_{2}\in \mathbb{F}_{q}^{\ast}}\zeta_p^{\Tr_p^q(-z_2t)}\sum\limits_{y\in \mathbb{F}_{q_2}}\zeta_{p}^{\Tr_p^{q_2}((z_{1}\alpha +z_{2}v)y)}\sum\limits_{x\in \mathbb{F}_{q_1}}\zeta_{p}^{\Tr_p^q(z_{1}f(x)+z_{2}\Tr_q^{q_1}(ux))}.}
\end{align}

(1) When $v\in \mathbb{F}_{q_2}\setminus \mathbb{F}_{q}^{\ast}\alpha$, then $\sum_{y\in \mathbb{F}_{q_2}}\zeta_{p}^{\Tr_p^{q_2}((z_{1}\alpha +z_{2}v)y)}=0$ and the desired conclusion is obtained.

(2) When $v\in \mathbb{F}_{q}^{\ast}\alpha$, i.e., $\alpha=zv$ for some $z\in \mathbb{F}_{q}^{\ast}$, \eqref{eq:N(u,v)} becomes
\begin{align}\label{eq:N(u,v):2}
|N(u,v,t)|
&=q^{s-2}+q^{s_2-2}\sum_{z_{1}\in \mathbb{F}_{q}^{\ast}}\zeta_p^{\Tr_p^q(z_1(-\beta+zt))}\sum_{x\in \mathbb{F}_{q_1}}\zeta_{p}^{\Tr_p^q(z_{1}f(x)-z_{1}\Tr_q^{q_1}(zux))}.
\end{align}

(2.1) If $u\notin \textrm{Im}(L_f)$, by Lemma~\ref{lem:6}, we have $|N(u,v,t)|=q^{s-2}$, the desired conclusion of (2.1) is obtained.

(2.2) If $u\in \textrm{Im}(L_f)$,
define $c=zu$, we have $x_{c}=zx_{u}$. By Lemma~\ref{lem:6} and Corollary~\ref{cor:1}, \eqref{eq:N(u,v):2} becomes
\begin{align}
&|N(u,v,t)|
=\left\{\begin{array}{ll}
    		q^{s-2}\Big(1+\varepsilon_f(p^*)^{-\frac{mR}{2}}\upsilon(\beta-zt+f(x_c))\Big), & 2|R, \\
    	q^{s-2}\Big(1+\varepsilon_f(p^*)^{-\frac{m(R-1)}{2}}\eta(\beta-zt+f(x_c))\Big),  & 2\not|R.
    	\end{array}
    	\right. \nonumber
\end{align}

\end{proof}

\begin{theorem}\label{thm:wd-e}  If $R$ is even, denote $\epsilon=\eta(-1)^{\frac{R}{2}}\varepsilon_f$,
then the code $C_{D}$ is a $[q^s-e,s]$ linear code over $\F_q$
with the weight distribution in Tables 1 and its complete weight enumerator is
\begin{align*}	
&\substack{CWE(C_D)=\omega_0^{q^{s-1}-e}+\Big(q^{s}-(q-1)q^{R}-1\Big)\omega_\rho^{q^{s-2}-e}\prod\limits_{\rho=1}^{q-1}\omega_\rho^{q^{s-2}}}\\		
&\substack{+q^{R-1}(q-1)(1+\upsilon(\beta)\epsilon q^{-\frac{R}{2}})\omega_0^{q^{s-2}\Big(1+(q-1)\epsilon q^{-\frac{R}{2}}\Big)-e}\prod\limits_{\rho=1}^{q-1} \omega_\rho^{q^{s-2}(1-\epsilon q^{-\frac{R}{2}})}}\\
&\substack{+q^{R-1}\omega_0^{q^{s-2}(1-\epsilon q^{-\frac{R}{2}})-e}\sum\limits_{h=1}^{q-1}\sum\limits_{l=1}^{q-1}\Big(1+\epsilon v(\omega_l-\beta)q^{-\frac{R}{2}}\Big)\prod\limits_{\rho=1}^{q-1} \omega_\rho^{q^{s-2}\Big(1+\epsilon v(\omega_l-\omega_h\omega_\rho) q^{-\frac{R}{2}}\Big)}}.
		\end{align*}
		\end{theorem}
\begin{table}
\centering
\caption{when $R$ is even}
\begin{tabular*}{10.5cm}{@{\extracolsep{\fill}}ll}
\hline
\textrm{Weight} $\gamma$ \qquad& \textrm{Multiplicity} $A_\gamma$   \\
\hline
0 & 1  \\
$q^{s-2}(q-1)$& $q^{s}-q^{R}(q-1)-1$  \\
$q^{s-2}(q-1)(1-\epsilon q^{-\frac{R}{2}})$ & $q^{R-1}(q-1)(1+\upsilon(\beta)\epsilon q^{-\frac{R}{2}})$  \\
$q^{s-2}\Big(q-1+\epsilon q^{-\frac{R}{2}}\Big)$ & $q^{R-1}(q-1)\Big(q-1-\upsilon(\beta)\epsilon q^{-\frac{R}{2}}\Big)$  \\
\hline
\end{tabular*}
\end{table}
\begin{proof} From Lemmas~\ref{lem:length} and \ref{lem:N(u,v,t)} we obtain the Hamming weight
\begin{align*}
  \mathrm{wt}(c_{(u,v)})&=q^{s-1}-N(u,v,0)\\ 
&=\left\{\begin{array}{ll}
    0,&\substack{u=v=0},\\
    q^{s-2}(q-1),& \substack{v\in \F_{q_2}\setminus\F_q^*\alpha \text{\ or}\ v\in \F_q^*\alpha \text{\ and}\ u\notin \Im{\ L_f}},\\
    	q^{s-2}(q-1)\Big(1-\epsilon q^{-\frac{R}{2}}\Big), &\substack{v\in \F_q^*\alpha,u\in\Im{\ L_f} \text{\ and}\ \beta+f(x_c)=0}, \\
    	q^{s-2}\Big(q-1+\epsilon q^{-\frac{R}{2}}\Big),  & \substack{v\in \F_q^*\alpha,u\in\Im{\ L_f} \text{\ and}\ \beta+f(x_c)\neq 0}.
    	\end{array}
    	\right.
\end{align*}
Thus,for any non-zero element $(u,v)\in\F$, we have $\wt(c_{(u,v)})>0$. So, the map: $\F\rightarrow C_{D}$ defined by $(u,v)\mapsto c_{(u,v)} $
is an isomorphism of linear spaces over $\F_q$. Hence, the dimension of the code $C_{D}$ in \eqref{defcode1} is equal to $s$. By Lemma \ref{lem:length}, we proved that the code $C_{D}$ is a $[p^{s-1}-e,s]$ linear code over $ \mathbb{F}_{q} $.

Define $$\gamma_1=q^{s-2}(q-1),\gamma_2=q^{s-2}(q-1)\Big(1-\epsilon q^{-\frac{R}{2}}\Big),\gamma_3=	q^{s-2}\Big(q-1+\epsilon q^{-\frac{R}{2}}\Big),$$
then the multiplicities $A_{\gamma_i}$ of codewords with weight $\gamma_i$ in $C_D$ is 
\begin{align*}\label{eq:1}
  &A_{\gamma_{1}} = \Big|\Big\{(u,v)\in\F|\wt(c_{(u,v)}) = \gamma_1\Big\}\Big|  \\
  &= \Big|\Big\{(u,v)\in\F|u\in\F_{q_1},v\in\F_{q_2}\setminus\F_q^*\alpha\Big\}\Big|+\Big|\Big\{(u,v)\in\F|u\notin\textrm{Im}(L_f),v\in\F_q^*\alpha\Big\}\Big|\\
&=\Big(q_1(q_2-(q-1))-1\Big)+(q_1-q^{R})(q-1)\\
&=q^{s}-q^{R}(q-1)-1,
\end{align*}
\begin{align*}
&A_{\gamma_{2}} = \Big|\Big\{(u,v)\in\F|\wt(c_{(u,v)}) = \gamma_2\Big\}\Big|  \\
  &= \Big|\Big\{(u,v)\in\F|v\in \F_q^*\alpha \text{\ and}\ \beta+f(x_c)=0,c=zu,z\in\F_q^*\Big\}\Big|\\
&= \Big|\Big\{(u,v)\in\F|v\in \F_q^*\alpha \text{\ and}\ f(x_u)=-\frac{\beta}{z^2},z\in\F_q^*\Big\}\Big|\\
&=(q-1)\frac{q^{s_1-1}+\upsilon(\beta)\epsilon q^{s_1-\frac{R+2}{2}}}{q^{s_1-R}}\\
&=q^{R-1}(q-1)(1+\upsilon(\beta)\epsilon q^{-\frac{R}{2}}),
\end{align*}
\begin{align*}
&A_{\gamma_{3}} = \Big|\Big\{(u,v)\in\F|\wt(c_{(u,v)}) = \gamma_3\Big\}\Big|  \\
  &= \Big|\Big\{(u,v)\in\F|v\in \F_q^*\alpha \text{\ and}\ \beta+f(x_c)\neq 0,c=zu,z\in\F_q^*\Big\}\Big|\\
&= \Big|\Big\{(u,v)\in\F|v\in \F_q^*\alpha \text{\ and}\ f(x_u)\neq -\frac{\beta}{z^2},z\in\F_q^*\Big\}\Big|\\
&=(q-1)\frac{q^{s_1}-(q^{s_1-1}+\upsilon(\beta)\epsilon q^{s_1-\frac{R+2}{2}})}{q^{s_1-R}}\\
&=q^{R-1}(q-1)(q-1-\upsilon(\beta)\epsilon q^{-\frac{R}{2}}).
\end{align*}

Let $\F_q=\{\omega_0,\omega_1,\cdots,\omega_{q-1}\}$. The complete weight enumerator $\omega[c_{(u,v)}]$ of the codeword $c_{(u,v)}$ is
\begin{align*}
    &\omega[c_{(u,v)}]=\omega_0^{N(u,v,0)-e}\omega_1^{N(u,v,\omega_1)}\cdots\omega_{q-1}^{N(u,v,\omega_{q-1})}\\
   &=\left\{\begin{array}{ll}
    \omega_0^{q^{s-1}-e},&\substack{u=v=0,}\\
    \omega_0^{q^{s-2}-e}\prod\limits_{\rho=1}^{q-1}\omega_\rho^{q^{s-2}},& \substack{v\in \F_{q_2}\setminus\F_q^*\alpha,\text{\ or}\ \alpha=zv, u\notin \Im{L_f},}\\
    \omega_0^{q^{s-2}\Big(1+\epsilon q^{-\frac{R}{2}}(q-1)\Big)-e}\prod\limits_{\rho=1}^{q-1} \omega_\rho^{q^{s-2}(1-\epsilon q^{-\frac{R}{2}})}, &\substack{\alpha=zv,\beta+z^2f(x_u)=0}, \\
    	\omega_0^{q^{s-2}(1-\epsilon q^{-\frac{R}{2}})-e}\prod\limits_{\rho=1}^{q-1} \omega_\rho^{q^{s-2}\Big(1+\epsilon v(\omega_l-z\omega_\rho) q^{-\frac{R}{2}}\Big)},  &\substack{\alpha=z v,\beta+z^2f(x_u)=\omega_l\neq 0.}
    	\end{array}
    	\right.
\end{align*}
Then the complete weight enumerator of the code $C_D$ 
can be desired.
\end{proof}

As special cases of Theorem~\ref{thm:wd-e}, we give the following three examples, which are verified by the Magma program.
\begin{example}
 Let $(q,s_{1},s_{2},\alpha,\beta)=(3,5,3,1,1)$,$f(x)=\Tr_{3}^{3^5}(2x^{10}+x^{2})$, we have $R=4$ and $\varepsilon_f =-1$. Then,the corresponding code $C_{D}$ has parameters $[2187,8,1377]$, the complete weight enumerator is
\begin{align*}&\omega_{0}^{2187}+6398\omega_{0}^{729}\omega_{1}^{729}\omega_{2}^{729}+60\omega_{0}^{567}\omega_{1}^{810}\omega_{2}^{810}+51\omega_{0}^{810}\omega_{1}^{567}\omega_{2}^{810}\\
    &+51\omega_{0}^{810}\omega_{1}^{810}\omega_{2}^{567}.
    \end{align*}
\end{example}

\begin{example}
 Let $(q,s_{1},s_{2},\alpha,\beta)=(3,2,4,1,1)$,$f(x)=\Tr_{3}^{3^2}(x^{2})$,by Corollary 1 in \cite{TXF17}, we have $R=2$ and $\varepsilon_f = 1$. Then, the corresponding code $C_{D}$ has parameters $[243,6,108]$ and the complete weight enumerator is 
\begin{align*}
  \omega_{0}^{243}+710\omega_{0}^{81}\omega_{1}^{81}\omega_{2}^{81}+4\omega_{0}^{135}\omega_{1}^{54}\omega_{2}^{54}+7\omega_{0}^{54}\omega_{1}^{135}\omega_{2}^{54}+7\omega_{0}^{54}\omega_{1}^{54}\omega_{2}^{135}.\end{align*}
\end{example}

\begin{example}
 Let $(q,s_{1},s_{2},\beta,\alpha)=(9,2,1,1+g,1)$,$f(x)=\mathrm Tr_{9}^{9^2}(x^{2})$, where $g$ is a primitive element of $\mathbb{F}_{9}$, then $R=2$ and $\epsilon=-1$. The corresponding code $C_{D}$ has parameters $[81,3,71]$ and the complete weight enumerator is $1+586z^{71}+80z^{72}+80z^{80}$.
\end{example}

\begin{theorem}\label{thm:wd-o}  
If $R$ is odd, let $\epsilon=\eta(-1)^{\frac{R-1}{2}}\varepsilon_f$.
Then the code $C_{D}$ is a $[q^{s-1}-e,s]$ linear code over $\F_q$
with the weight distribution in Tables \rm{2 \&  3} and its complete weight enumerator is
\begin{align*}
&\substack{CWE(C_D)=\omega_0^{q^{s-1}-e}+(q^{s}-(q-1)q^{R}-1)\omega_0^{q^{s-2}-e}\prod\limits_{\rho=1}^{q-1}\omega_\rho^{q^{s-2}}}\\
&\substack{+q^{R-1}\Big(1+\epsilon\eta(-\beta)q^{-\frac{R-1}{2}}\Big)\omega_0^{q^{s-2}-e}\sum\limits_{l=1}^{q-1}\prod\limits_{\rho=1}^{q-1}\omega_\rho^{q^{s-2}\Big(1+\epsilon\eta(-\omega_l \omega_t)q^{-\frac{R-1}{2}}\Big)}}\\			&\substack{+q^{R-1}\sum\limits_{l=1}^{q-1}\sum\limits_{\substack{\theta\in\F_q^*,\\\eta(\theta)=\epsilon}}\Big(1+\eta(1-\frac{\beta}{\theta})q^{-\frac{R-1}{2}}\Big)\omega_0^{q^{s-2}(1+ q^{-\frac{R-1}{2}})-e}\prod\limits_{\rho=1}^{q-1}\omega_\rho^{q^{s-2}\Big(1+\epsilon\eta(\theta-\omega_l \omega_\rho)q^{-\frac{R-1}{2}}\Big)}}\\		&\substack{+q^{R-1}\sum\limits_{l=1}^{q-1}\sum\limits_{\substack{\theta\in\F_q^*,\\\eta(\theta)=-\epsilon}}\Big(1-\eta(1-\frac{\beta}{\theta})q^{-\frac{R-1}{2}}\Big)\omega_0^{q^{s-2}(1- q^{-\frac{R-1}{2}})-e}\prod\limits_{\rho=1}^{q-1}\omega_\rho^{q^{s-2}\Big(1+\epsilon\eta(\theta-\omega_l \omega_\rho)q^{-\frac{R-1}{2}}\Big)}}.
		\end{align*}
\end{theorem}
\begin{table}\label{tab:wd:o}
\centering
\caption{when $R$ is odd and $\beta = 0$}
\begin{tabular*}{9cm}{@{\extracolsep{\fill}}ll}
\hline
\textrm{Weight} $\omega$ & \textrm{Multiplicity} $A_\omega$   \\
\hline
$\substack{0}$ & $\substack{1}$  \\
				$\substack{q^{s-2}(q-1)}$ & $\substack{q^{s}-q^{R-1}(q-1)^2-1}$\\
		$\substack{q^{s-2}\Big(q-1-q^{-\frac{R-1}{2}}\Big)}$&$\substack{\frac{1}{2}q^{R-1}(q-1)^2(1+q^{-\frac{R-1}{2}})}$  \\	$\substack{q^{s-2}\Big(q-1+q^{-\frac{R-1}{2}}\Big)}$ & $\substack{\frac{1}{2}q^{R-1}(q-1)^2(1-q^{-\frac{R-1}{2}})}$  \\
\hline
\end{tabular*}
\end{table}
\begin{table}\label{tab:wd:o}
\centering
\caption{when $R$ is odd and $\beta \neq 0$}
\begin{tabular*}{9cm}{@{\extracolsep{\fill}}ll}
\hline
\textrm{Weight} $\omega$ & \textrm{Multiplicity} $A_\omega$   \\
\hline
$\substack{0}$ & $\substack{1}$  \\
				$\substack{q^{s-2}(q-1)}$ & $\substack{q^{s}-q^{R-1}(q-1)^2+\epsilon\eta(-\beta)q^{\frac{R-1}{2}}(q-1)-1}$\\
		$\substack{q^{s-2}\Big(q-1-q^{-\frac{R-1}{2}}\Big)}$&$\substack{\frac{1}{2}q^{R-1}(q-1)^2+\frac{1}{2}q^{\frac{R-1}{2}}(q-1)(-1-\epsilon\eta(-\beta))}$  \\	$\substack{q^{s-2}\Big(q-1+q^{-\frac{R-1}{2}}\Big)}$ & $\substack{\frac{1}{2}q^{R-1}(q-1)^2-\frac{1}{2}q^{\frac{R-1}{2}}(q-1)(-1+\epsilon\eta(-\beta))}$  \\
\hline
\end{tabular*}
\end{table}
\begin{proof} Similarly,from Lemmas~\ref{lem:length} and \ref{lem:N(u,v,t)} we obtain the Hamming weight
\begin{align*}
  &\mathrm{wt}(c_{(u,v)})=q^{s-1}-N(u,v,0)\\ 
&=\left\{\begin{array}{ll}
    0,&u=v=0,\\
    q^{s-2}(q-1),& v\in \F_{q_2}\setminus\F_q^*\alpha \text{\ or}\ v\in \F_q^*\alpha \text{\ and}\ u\notin \Im{\ L_f},\\
    & \text{\ or}\ v\in \F_q^*\alpha,u\in\Im{\ L_f},\beta+f(x_c)=0,\\
    	q^{s-2}\Big(q-1-q^{-\frac{R-1}{2}}\Big), &v\in \F_q^*\alpha,u\in\Im{\ L_f}, \eta(\beta+f(x_c))=\epsilon, \\
    	q^{s-2}\Big(q-1+q^{-\frac{R-1}{2}}\Big), &v\in \F_q^*\alpha,u\in\Im{\ L_f},\eta(\beta+f(x_c))=-\epsilon.
    	\end{array}
    	\right.
\end{align*}
Define $$\gamma_1=q^{s-2}(q-1),\gamma_2=	q^{s-2}\Big(q-1-q^{-\frac{R-1}{2}}\Big),\gamma_3=	q^{s-2}\Big(q-1+q^{-\frac{R-1}{2}}\Big),$$
then the multiplicity $A_{\gamma}$ of codewords with weight $\gamma_i$ in $C_D$ is 
\begin{align*}
  &A_{\gamma_{1}} = \Big|\Big\{(u,v)\in\F|\wt(c_{(u,v)}) = \gamma_1\Big\}\Big| \\
  &= \Big|\Big\{(u,v)\in\F|u\in\F_{q_1},v\in\F_{q_2}\setminus\F_q^*\alpha\Big\}\Big|\\
  &+\Big|\Big\{(u,v)\in\F|u\notin\textrm{Im}(L_f),\alpha = zv,z\in\F_q^*\Big\}\Big|\\
  &+\Big|\Big\{(u,v)\in\F|u\in\Im{\ L_f},\alpha = zv,z\in\F_q^*, \beta+z^2f(x_u)=0\Big\}\Big|\\
&=\Big(q_1(q_2-(q-1))-1\Big)+(q_1-q^{R})(q-1)\\
&+(q-1)\frac{q^{s_1-1}+\epsilon \eta(-\beta )q^{s_1-\frac{R+1}{2}}}{q^{s_1-R}}\\
&=q^{s}-q^{R-1}(q-1)^2+\epsilon\eta(-\beta)q^{\frac{R-1}{2}}(q-1)-1,
\end{align*}
\begin{align*}
&A_{\gamma_{2}} = \Big|\Big\{(u,v)\in\F|\wt(c_{(u,v)}) = \gamma_2\Big\}\Big|  \\
  &= \Big|\Big\{(u,v)\in\F|\alpha = zv,u\in\Im{\ L_f},\eta(\beta+z^2f(x_u))=\epsilon,z\in\F_q^*\Big\}\Big|\\
&=(q-1)\sum\limits_{\substack{\theta\in\F_q^*,\\\eta(\theta)=\epsilon}}\frac{q^{s_1-1}+\epsilon \eta(\theta-\beta)q^{s_1-\frac{R+1}{2}}}{q^{s_1-R}}\\
&=\frac{1}{2}q^{R-1}(q-1)^2+\epsilon q^{\frac{R-1}{2}}(q-1)\sum\limits_{\substack{\theta\in\F_q^*,\\\eta(\theta)=\epsilon}} \eta(\theta -\beta).\end{align*}

If $\beta = 0$, it's easy to see that $A_{\gamma_{2}}=\frac{1}{2}q^{R-1}(q-1)^2(1+q^{-\frac{R-1}{2}})$.

If $\beta \neq 0$, let $I_n(a)=\sum\limits_{x\in\F_q}\eta(a+x^n), a\in\F_q^*$. By Theorem 5.50 in \cite{LN97}, we know that $I_2(a)=-1$ for any $a\in\F_q^*$. Then when $\epsilon\eta(-\beta)=1$, we have 

\begin{align*}
&\sum\limits_{\substack{\theta\in\F_q^*,\\\eta(\theta)=\epsilon}} \eta(\theta -\beta)=\eta(-\beta)\sum\limits_{\substack{\theta\in\F_q^*,\\\eta(\theta)=\eta(-\beta)\epsilon}}\eta(1+\theta)=\frac{1}{2}\epsilon\sum\limits_{\theta\in\F_q^*}\eta(1+\theta^2)\\
&= \frac{1}{2}\epsilon\Big(\sum\limits_{\theta\in\F_q}\eta(1+\theta^2)-1\Big)= \frac{1}{2}\epsilon(I_2(1)-1)=-\epsilon.
\end{align*}

When $\epsilon\eta(-\beta)=-1$, we have
\begin{align*}
&\sum\limits_{\substack{\theta\in\F_q^*,\\\eta(\theta)=\epsilon}} \eta(\theta -\beta)=\eta(-\beta)\sum\limits_{\substack{\theta\in\F_q^*,\\ \eta(\theta)=\eta(-\beta)\epsilon}}\eta(1+\theta)=\frac{1}{2}\epsilon\sum\limits_{\theta\in\F_q^*}\eta(\gamma+\theta^2) \\
&= \frac{1}{2}\epsilon\Big(\sum\limits_{\theta\in\F_q}\eta(\gamma+\theta^2)+1\Big)= \frac{1}{2}\epsilon(I_2(\gamma)+1)=0,
\end{align*}
where $\gamma$ is some fixed non-square element in $\mathbb{F}_q^{*}$. From the above analysis, we have
\begin{align*}
&A_{\gamma_{2}}
=\frac{1}{2}q^{R-1}(q-1)^2+\epsilon\eta(-\beta) q^{\frac{R-1}{2}}(q-1)\sum\limits_{\substack{\theta\in\F_q^*,\\\eta(\theta)=\eta(-\beta)\epsilon}} \eta(1+\theta)\\
&=\frac{1}{2}q^{R-1}(q-1)^2+\frac{1}{2} q^{\frac{R-1}{2}}(q-1)(-1-\epsilon\eta(-\beta)),
\end{align*}

By similar calculations and analysis, we can obtain $$ A_{\gamma_{3}} =\left\{\begin{array}{ll} 
\frac{1}{2}q^{R-1}(q-1)^2(1-q^{-\frac{R-1}{2}}),& \textrm{if}\ \beta=0,\\
\frac{1}{2}q^{R-1}(q-1)^2-\frac{1}{2} q^{\frac{R-1}{2}}(q-1)(-1+\epsilon\eta(-\beta)), & \textrm{if}\ \beta \neq 0.
\end{array}
    	\right.$$


The complete weight enumerator $\omega[c_{(u,v)}]$ of the codeword $c_{(u,v)}$ is
\begin{align*}
    &\omega[c_{(u,v)}]=\omega_0^{N(u,v,0)-e}\omega_1^{N(u,v,\omega_1)}\cdots\omega_{q-1}^{N(u,v,\omega_{q-1})}\\
   & =\left\{\begin{array}{ll}
    \omega_0^{q^{s-1}-e},& u=v=0,\\
    \omega_0^{q^{s-2}-e}\prod\limits_{\rho=1}^{q-1}\omega_\rho^{q^{s-2}},& v\in \F_{q_2}\setminus\F_q^*\alpha,\text{\ or}\\& v\in \F_q^*\alpha, u\notin \Im{L_f},\\
    \omega_0^{q^{s-2}-e}\prod\limits_{\rho=1}^{q-1}\omega_\rho^{q^{s-2}\Big(1+\epsilon\eta(-z \omega_\rho)q^{-\frac{R-1}{2}}\Big)}, &\alpha=z v,u\in\Im{\ L_f},\\
    &\text{\ and}\ \beta+f(x_c)=0, \\
    	\omega_0^{q^{s-2}(1+ q^{-\frac{R-1}{2}})-e}\prod\limits_{\rho=1}^{q-1}\omega_\rho^{q^{s-2}\Big(1+\epsilon\eta(\theta-z \omega_\rho)q^{-\frac{R-1}{2}}\Big)},  & \beta+z^2f(x_u)=\theta,\\
    	&\eta(\theta)=\epsilon,\alpha=z v,\\
    		\omega_0^{q^{s-2}(1- q^{-\frac{R-1}{2}})-e}\prod\limits_{\rho=1}^{q-1}\omega_\rho^{q^{s-2}\Big(1+\epsilon\eta(\theta-z \omega_\rho)q^{-\frac{R-1}{2}}\Big)},  & \beta+z^2f(x_u)=\theta,\\
    	&\eta(\theta)=-\epsilon,\alpha=z v.
    	\end{array}
    	\right.
\end{align*}
Then the complete weight enumerator of the code $C_D$ can be obtained.
\end{proof}

As special cases of Theorem~\ref{thm:wd-o}, we give the following three examples, which are verified by Magma programs.

\begin{example}
 Let $(q,s_{1},s_{2},\beta,\alpha)=(5,3,2,1,1)$,$f(x)=\Tr_{5}^{5^3}(\theta x^{2})$, where $\theta$ is a primitive element of $\mathbb{F}_{5}^{3}$, by Corollary 1 in \cite{TXF17}, we have $R=3$ and $\varepsilon_{f}=-1$. Then,the corresponding code $C_{D}$ has parameters $[625,5,475]$, the complete weight enumerator is  
 \begin{align*}
&\omega_{0}^{625}+2624\omega_{0}^{125}\omega_{1}^{125}\omega_{2}^{125}\omega_{3}^{125}\omega_{4}^{125}+40\omega_{0}^{125}\omega_{1}^{150}\omega_{2}^{100}\omega_{3}^{100}\omega_{4}^{150}\\
&+40\omega_{0}^{125}\omega_{1}^{100}\omega_{2}^{150}\omega_{3}^{150}\omega_{4}^{100}+50\omega_{0}^{150}\omega_{1}^{125}\omega_{2}^{150}\omega_{3}^{100}\omega_{4}^{100}\\
&+50\omega_{0}^{150}\omega_{1}^{150}\omega_{2}^{100}\omega_{3}^{125}\omega_{4}^{100}+50\omega_{0}^{150}\omega_{1}^{100}\omega_{2}^{125}\omega_{3}^{100}\omega_{4}^{150}\\
&+50\omega_{0}^{150}\omega_{1}^{100}\omega_{2}^{100}\omega_{3}^{150}\omega_{4}^{125}+55\omega_{0}^{100}\omega_{1}^{150}\omega_{2}^{125}\omega_{3}^{150}\omega_{4}^{100}\\
&+55\omega_{0}^{100}\omega_{1}^{125}\omega_{2}^{100}\omega_{3}^{150}\omega_{4}^{150}+55\omega_{0}^{100}\omega_{1}^{150}\omega_{2}^{150}\omega_{3}^{100}\omega_{4}^{125}\\
&+55\omega_{0}^{100}\omega_{1}^{100}\omega_{2}^{150}\omega_{3}^{125}\omega_{4}^{150}.
\end{align*}

 \end{example}

\begin{example}
 Let $(q,s_{1},s_{2},\beta,\alpha)=(3,3,4,1,1)$,$f(x)=\Tr_{3}^{3^3}(\theta x^{2})$, where $\theta$ is a primitive element of $\mathbb{F}_{3}^{3}$, by Corollary 1 in \cite{TXF17}, we have $R=3$ and $\varepsilon_f = -1$. Then, the corresponding code $C_{D}$ has parameters $[729,7,405]$ and the complete weight enumerator is  \begin{align*}
  & \omega_{0}^{729}+2132\omega_{0}^{243}\omega_{1}^{243}\omega_{2}^{243}+6\omega_{0}^{243}\omega_{1}^{162}\omega_{2}^{324}+6\omega_{0}^{243}\omega_{1}^{324}\omega_{2}^{162}\\&+9\omega_{0}^{324}\omega_{1}^{243}\omega_{2}^{162}+9\omega_{0}^{324}\omega_{1}^{162}\omega_{2}^{243}+12\omega_{0}^{162}\omega_{1}^{324}\omega_{2}^{243}+12\omega_{0}^{162}\omega_{1}^{243}\omega_{2}^{324}. 
\end{align*}
\end{example}

\begin{example}
Let $(q,s_1,s_2,\beta,\alpha)=(3,4,1,0,1)$ and $f(x)=\mathrm{Tr}_3^{3^4}(x^{2})-\frac{1}{4}(\mathrm{Tr}_3^{3^4}(x))^{2}$. By Corollary 2 in \cite{TXF17}, we have
$\varepsilon_{f}=1$ and $R=3$. Then, the corresponding code $C_{D}$ has parameters $[80,5,45]$ and the weight enumerator
$1+24x^{45}+206x^{54}+12x^{63}$.
\end{example}

\begin{remark}
Let $\omega_{min}$ and $\omega_{max}$ be the minimum and maximum Hamming weights of nonzero codewords in $C_D$.
By Theorem~\ref{thm:wd-e} and Theorem \ref{thm:wd-o}, if $R\geq 3$, we always have 
\begin{center}
    $\cfrac{\omega_{min}}{\omega_{max}} > \cfrac{q-1}{q}$,
\end{center}
by the results in \cite{AB98} and \cite{YD06}, the codewords in $C_{D}$ are minimal codewords, so they can be employed to obtain secret sharing schemes.
\end{remark}

\section{The weight hierarchies of the presented linear codes}

In this section, we give the weight hierarchy of $C_{D}$ in \eqref{defcode1}.

For the linear codes $C_D$ defined in \eqref{defcode1}, a general formula may be employed to calculate the generalized Hamming weight $d_r(C_D)$. It is presented in the following proposition, which is the bivariate form of \cite[Proposition 2.1]{HLL24} and \cite[Theorem 1]{LF18}.

\begin{proposition}\label{pro:d_r}
For each $ r $ and $ 1\leq r \leq s$, if the dimension of $ C_{D} $
is $ s$, then \begin{align}
 d_{r}(C_{D})&=n-\max\Big\{|H_r^\perp\cap D|: H_r \in [\mathbb{F},r]_{q}\Big\}.  
       \end{align}
where $H^{\perp}=\{(x,y)\in \mathbb{F}:\Tr_{q}^{q^{s_1}}(ux)+\Tr_{q}^{q^{s_2}}(vy)=0,  (u,v)\in H \}$.
\end{proposition}

By Theorem~\ref{thm:wd-e} and Theorem~\ref{thm:wd-o}, we know that the dimension of the code $C_{D}$ defined in \eqref{defcode1} is $s$.
Let $H_r$ be an $r$-dimensional subspace of $\mathbb{F}$. On one hand, let $
N(H_r)=\Big\{(x,y)\in \mathbb{F}: f(x)+\mathrm{Tr}_q^{q_2}(\alpha y)=\beta, \mathrm{Tr}_{q}^{q_1}(ux)+\mathrm{Tr}_{q}^{q_2}(vy)=0, \forall (u,v)\in H_r \Big\}$.
Hence, by Lemma~\ref{lem:length} and Proposition~\ref{pro:d_r}, we have
\begin{equation}\label{eq:d_r:3}
     d_{r}(C_{\D})=q^{s-1}-\max\Big\{N(H_r)|: H_r \in [\mathbb{F},r]_{q}\Big\}.
\end{equation}

\begin{lemma} \label{lem:d_r:2}
Let $\alpha \in\F_{q_2}^*$ and $f$ be a quadratic form defined in \eqref{eq:f} with the sign $\varepsilon_f$ and the rank $R$. $H_r$ and $N(H_r)$ are defined as above. We have the following.
\begin{itemize}
  \item[(1)] If $\alpha\notin  \Prj_{2}(H_r)$, then $N(H_r)=q^{s-(r+1)}$.
  \item[(2)] If $\alpha\in  \Prj_{2}(H_r)$, then 
  \item[(2.1)] when $R$ is even, we have
  \[
    		|N(H_r)|=
    		q^{s-(r+1)}\big[1+\varepsilon_f(p^*)^{-\frac{mR}{2}}\sum\limits_{(u,-\alpha)\in H_{r}} \upsilon(\beta+f(x_u))\big].
    	\]
    \item[(2.2)] when $R$ is odd, we have
    \[
    		|N(H_r)|=
    	q^{s-(r+1)}\big[1+\varepsilon_f(p^*)^{-\frac{m(R-1)}{2}}\sum\limits_{(u,-\alpha)\in H_{r}}\eta(\beta+f(x_u))\big].
    	\]
\end{itemize}
Here $ \Prj_{2}$ is the second projection from $\mathbb{F}$ to $\mathbb{F}_{q_2}$ defined by $(x,y)\mapsto y$.
\end{lemma}

\begin{proof}
By the orthogonal property of additive characters, we have
\begin{align*}
&\substack{q^{r+1}|N(H_r)|
    		=\sum\limits_{(x,y)\in \mathbb{F}}\sum\limits_{z\in \mathbb{F}_{q}}\zeta_p^{\mathrm{Tr}_{p}^{q}\big(zf(x)+\mathrm{Tr}_{q}^{q_2}(z\alpha y)-z\beta\big)}\sum\limits_{(u,v)\in H_r}\zeta_p^{\mathrm{Tr}_{p}^{q}\big(\mathrm{Tr}_q^{q_1}(ux)+\mathrm{Tr}_q^{q_2}(vy)\big)}}
    		\\
    		&=\substack{\sum\limits_{(x,y)\in \mathbb{F}}\sum\limits_{(u,v)\in H_r}\zeta_p^{\mathrm{Tr}_{p}^{q}\big(\mathrm{Tr}_q^{q_1}(ux)
    		+\mathrm{Tr}_q^{q_2}(vy)\big)}}\\
    		&\substack{+\sum\limits_{(x,y)\in \mathbb{F}}\sum\limits_{z\in \mathbb{F}_{q}^{\ast}}\zeta_p^{-\mathrm{Tr}_{p}^{q}(z\beta)}\sum\limits_{(u,v)\in H_r}\zeta_p^{\mathrm{Tr}_{p}^{q}\big(zf(x)+\mathrm{Tr}_q^{q_1}(ux)+\mathrm{Tr}_q^{q_2}(vy)+\mathrm{Tr}_{q}^{q_2}(z\alpha y)\big)}}\\
    		&\substack{=q^s+\sum\limits_{(x,y)\in \mathbb{F}}\sum\limits_{z\in \mathbb{F}_{q}^{\ast}}\zeta_p^{-\mathrm{Tr}_{p}^{q}(z\beta)}\sum\limits_{(u,v)\in H_r}\zeta_p^{\mathrm{Tr}_{p}^{q}\big(zf(x)+\mathrm{Tr}_q^{q_1}(ux)+\mathrm{Tr}_q^{q_2}(vy)+\mathrm{Tr}_{q}^{q_2}(z\alpha y)\big)}}\\
    		&\substack{=q^s+\sum\limits_{(u,v)\in H_r}\sum\limits_{z\in \mathbb{F}_{q}^{\ast}}\zeta_p^{-\mathrm{Tr}_{p}^{q}(z\beta)}\sum\limits_{x\in \mathbb{F}_{q_1}}\zeta_p^{\mathrm{Tr}_{p}^{q}\big(zf(x)+z\mathrm{Tr}_q^{q_1}(\frac{u}{z}x)\big)}\sum\limits_{y\in \mathbb{F}_{q_2}}\zeta_p^{\mathrm{Tr}_{p}^{q}\big(z(\mathrm{Tr}_q^{q_2}(\frac{v}{z}y)+\mathrm{Tr}_{q}^{q_2}(\alpha y))\big)}}
    		\\&\substack{=q^s+\sum\limits_{(u,v)\in H_r}\sum\limits_{z\in \mathbb{F}_{q}^{\ast}}\zeta_p^{-\mathrm{Tr}_{p}^{q}(z\beta)}\sum\limits_{x\in \mathbb{F}_{q_1}}\zeta_p^{\mathrm{Tr}_{p}^{q}\big(zf(x)+z\mathrm{Tr}_q^{q_1}(ux)\big)}\sum\limits_{y\in \mathbb{F}_{q_2}}\zeta_p^{\mathrm{Tr}_{p}^{q_2}\big(zy(v+\alpha)\big)}.}
    	\end{align*}
    	If $\alpha\notin  \Prj_{2}(H_r)$, then by $\displaystyle\sum_{y\in \mathbb{F}_{q_2}}\zeta_p^{\mathrm{Tr}_{p}^{q_2}\big(zy(v+\alpha)\big)}=0$, we obtain $|N(H_r)|=q^{s-(r+1)}$. \\
    	If $\alpha \in \mathrm{Prj}_2(H_r)$, by Lemma~\ref{lem:6}, we have
    	\begin{align*}
    		&\substack{q^{r+1}|N(H_r)|
    		=q^s+q_{2}\sum\limits_{(u,-\alpha)\in H_{r}}\sum\limits_{z\in \mathbb{F}_{q}^{\ast}}\zeta_p^{-\mathrm{Tr}_{p}^{q}(z\beta)}\sum\limits_{x\in \mathbb{F}_{q_1}}\zeta_p^{\mathrm{Tr}_{p}^{q}\big(zf(x)+z\mathrm{Tr}_q^{q_1}(ux)\big)}}
    		\\&\substack{=q^s+q_{2}\sum\limits_{(u,-\alpha)\in H_{r}}\sum\limits_{z\in \mathbb{F}_{q}^{\ast}}\zeta_p^{-\mathrm{Tr}_{p}^{q}(z\beta)}\Big(\sum\limits_{x\in \mathbb{F}_{q_1}}\zeta_p^{\mathrm{Tr}_{p}^{q}(zf(x)+z\mathrm{Tr}_q^{q_1}(ux))}\Big)}
    		\\&=\left\{\begin{array}{ll}
    			\substack{q^s}, & \substack{ u\notin \mathrm{Im}(L_f),} \\
    			\substack{q^s+\varepsilon_f q^{s}(p^*)^{-\frac{mR}{2}}\sum\limits_{(u,-\alpha)\in H_{r}}\sum\limits_{z\in \mathbb{F}_{q}^{\ast}}\zeta_p^{-\mathrm{Tr}_p^q(z(\beta+f(x_u)))}, } &\substack{ 2|R,u\in \mathrm{Im}(L_f),}\\
    			\substack{q^s+(-1)^{m-1}\eta(-1)\varepsilon_f q^{s}(p^*)^{-\frac{mR}{2}}\sum\limits_{(u,-\alpha)\in H_{r}}\sum\limits_{z\in \mathbb{F}_{q}^{\ast}}\eta(z)\zeta_p^{-\mathrm{Tr}_p^q(z(\beta+f(x_u)))},}  & \substack{2\not|R,u\in \mathrm{Im}(L_f).}
    		\end{array}
    		\right.
    	\end{align*}
    
So, the desired result follows from Lemma~\ref{cor:1}. Thus, we complete the proof.
\end{proof}

On the other hand, since the dimension of $H_r^\perp$ is $s-r$, by Proposition \ref{pro:d_r}, we give another general formula,
that is
\begin{align}\label{eq:d_r:2}
        d_{r}(C_{D})&=n-\max\Big\{|H_{s-r}\cap D|: H_{s-r} \in [\mathbb{F},s-r]_{q}\Big\}. 
\end{align}

From the orthogonality of exponential sums and Lemma \ref{lem:1}, we obtain the following Lemma.
\begin{lemma}\label{lem:dh}
Let $f(x)$ be quadratic forms over $\E$ defined in \eqref{eq:f} and $D$ be defined as \eqref{set:D1}. Let $H=H_1\times H_2$ be an $r$-dimensional subspace of $\F$, where $H_1$ (resp $H_2$) is an $r_1$ (resp $r_2$) -dimensional subspace of $\F_{q_1}$ (resp $\F_{q_2}$) and $r=r_1+r_2$, defining $\Big|D\cap H\Big|=\#\{(u,v)\in H:f(x)+\Tr_q^{q_2}(\alpha y)=\beta\}$. Then, we have
the following.

\rm{(1)} When $\Tr_q^{q_2}(\alpha H_2)\neq\{0\}$, we have $\Big|D\cap H\Big|=q^{r-1}$.

 \rm{(2)} When $\Tr_q^{q_2}(\alpha H_2)=\{0\}$, we have
  
  \begin{equation*}
  \Big|D\cap H\Big|=q^{r_2}\Big|\overline{D}_\beta\cap H_1\Big|=\left\{\begin{array}{ll}
q^{r-1}+\upsilon(\beta)\eta((-1)^{\frac{R_H}{2}})\varepsilon_{H}q^{r-\frac{R_H+2}{2}},  &\textrm{ $2|R_H$}, \\
q^{r-1}+\eta((-1)^{\frac{R_H-1}{2}}\beta)\varepsilon_{H}q^{r-\frac{R_H+1}{2}},  &\textrm{$2\not|R_H$}.
\end{array}
\right.
  \end{equation*}
\end{lemma}

\begin{corollary}\label{cor:dh} Let $f(x)$ be quadratic forms over $\E$ defined in \eqref{eq:f} and $D$ be defined as \eqref{set:D1} and $H$ is an $r$-dimensional subspace of $\F$. If $\beta\neq 0$, then
$$\Big|D\cap H\Big| \leq 2q^{r-1}.$$
\end{corollary}
\begin{proof} Define $H_1=\Prj_1(H),H_2=\Prj_2(H)$,we have $H\subseteq H_1\times H_2$.

If $\Tr_q^{q_2}(\alpha H_2)=\{0\}$, we decompose $H$ as follows 
\begin{equation*}\label{eq:sd1} H=(J_1\times J_2) \bigoplus H_3,\end{equation*}
which satisfies that \begin{align*}
&J_1=\Big\langle a_1,a_2,\cdots,a_k\Big\rangle,\\
&J_2=\Big\langle b_1,b_2,\cdots,b_l\Big\rangle,\\
&H_3=\Big\langle(\mu_1,\nu_1),(\mu_2,\nu_2),\cdots,(\mu_t,\nu_t)\Big\rangle,
\end{align*}
and $a_1,a_2,\cdots,a_k,\mu_1,\mu_2,\cdots,\mu_t$ (resp $b_1,b_2,\cdots,b_l,\nu_1,\nu_2,\cdots,\nu_t$) are $\F_q$-basis of $H_1$ (resp $H_2$), then by Lemma~\ref{lem:1} we have $\Big|D\cap H\Big|=q^l\Big|\overline{D}_{\beta}\cap H_1 \Big| \leq 2q^{r-1}$.  

If $\Tr_q^{q_2}(\alpha H_2)\neq\{0\}$, we can decompose $H$ as follows
\begin{equation*}\label{eq:sd2} H=\Big\langle(a_0,b_0)\Big\rangle \bigoplus H',\end{equation*}
where $\Tr_q^{q_2}(\alpha b_0)=1$ and $\Tr_q^{q_2}(\alpha \Prj_2(H')) = \{0\}$, 
then we have 
\begin{align*}
   \Big|D\cap H\Big|&=\Big|D\cap \Big\langle(a_0,b_0)\Big\rangle\Big|+\Big|D\cap (H'\setminus\{(0,0)\})\Big|\\
   &+ \sum_{k\in\F_q^*}\Big|D\cap (k(a_0,b_0)+H'\setminus\{(0,0)\})\Big|\\
   &\leq 1 + 2q^{r-2} + (q-1)q^{r-2}\\
   &=(q+1)q^{r-2}+1 \\
   &< 2q^{r-1}. 
\end{align*}

Therefore the desired conclusion is obtained.
\end{proof}

\begin{remark}
From the proof of Corollary~\ref{cor:dh}, if $R_{H_1}=1,\varepsilon_{H_1}=\eta(\beta)$ and $\Tr_q^{q_2}(\alpha H_2)=\{0\}$, then the maximum value of $\Big|D\cap H\Big|$ is equals to $2q^{r-1}$.
\end{remark}

The following two lemmas are the generalized forms of \cite[Lemma 2.5 and Lemma 2.6]{HLL24}, and their proofs are also similar, we omit them.

\begin{lemma}[{\cite[Lemma 2.5]{HLL24}}]\label{lem:2}
Let $f$ be a quadratic form over $\E$ with the rank $R\geq 3$. There exists an $e_0$-dimensional subspace
$H $ of \ $ \E$ such that $\E^{\perp_f}\subseteq H$ and $f(x)=0$ for any $x\in H$, where
$$
e_0=\left\{\begin{array}{ll}
s_1-\frac{R+1}{2}, & \textrm{if $R$ is odd}, \\
s_1-\frac{R}{2}, & \textrm{if $R$ is even and $\varepsilon_{f}=\eta(-1)^{\frac{R}{2}}$},\\
s_1-\frac{R+2}{2}, & \textrm{if $R$ is even and $\varepsilon_{f}=-\eta(-1)^{\frac{R}{2}}$}.
\end{array}
\right.
$$
\end{lemma}


\begin{lemma}\label{lem:2-2}
Let $f$ be a quadratic form over $\E$ with the rank $R\geq 3$. For each $a\in \F_q^*$, there exists an $(s_1-R+l_0)$-dimensional subspace
$H$ of \ $ \E$ such that $R_H=1, \varepsilon_{H}=\eta(a)$, where $
l_0=\frac{R-1}{2} $,if $R$ is odd; $
l_0=\frac{R}{2} $,if $R$ is even.
\end{lemma}

By Theorem 3.8 and its proof in \cite{HLL24}, similarly, we have the following theorem.
\begin{theorem}[{\cite[Theorem 3.8]{HLL24}}]\label{thm:wh} Let $\alpha \in \F_{q_2}^*,\beta = 0$ and $f$ be a homogeneous quadratic function defined in \eqref{eq:f} with the sign $\varepsilon_f$ and the rank $R\geq 3$. Let $e_0$ be defined as in Lemma~\ref{lem:2}. Then we have the following
\begin{itemize}
\item[(1)] When $0< r\leq s_1-e_0$, we have
$$
d_{r}(C_{\D})=\left\{\begin{array}{ll}
q^{s-1}-q^{s-1-r}-q^{s-1-\frac{R+1}{2}}, \textrm{if $2\nmid R$\ }, \\
q^{s-1}-q^{s-1-r}-(q-1)q^{s-1-\frac{R+2}{2}}, \textrm{if $2\mid R$, $\varepsilon_{f}=\eta(-1)^{\frac{R}{2}}$},\\
q^{s-1}-q^{s-1-r}-q^{s-1-\frac{R+2}{2}},  \textrm{if $2\mid R$, $\varepsilon_{f}=-\eta(-1)^{\frac{R}{2}}$}.
\end{array}
\right.
$$
\item[(2)]
When $s_1-e_0+1 \leq r \leq s$, we have
$$
d_{r}(C_{\D})=q^{s-1}-q^{s-r}.
$$
\end{itemize}
\end{theorem}

\begin{theorem} Let $\alpha \in \F_{q_2}^*,\beta \neq 0$ and $f$ be a homogeneous quadratic function defined in \eqref{eq:f} with the sign $\varepsilon_f$ and the rank $R\geq 3$. Let $l_0$ be defined as in Lemma~\ref{lem:2-2}. Then we have the following. \begin{itemize}
\item[(1)] When $0< r\leq l_0$, we have
$$
d_{r}(C_{D})=\left\{\begin{array}{ll}
q^{s-1}-q^{s-1-r}-(q-1)q^{s-1-\frac{R+2}{2}},\textrm{if $2\mid R$, $\varepsilon_{f}=\eta(-1)^{\frac{R}{2}}$},\\
q^{s-1}-q^{s-1-r}-q^{s-1-\frac{R+2}{2}}, \textrm{if $2\mid R$, $\varepsilon_{f}=-\eta(-1)^{\frac{R}{2}}$},\\
q^{s-1}-q^{s-1-r}-q^{s-1-\frac{R+1}{2}}, \textrm{if $2\nmid R$\ }.
\end{array}
\right.
$$
\item[(2)] When $l_0< r\leq s-1$, we have
\item[(2.1)] if $R $ is even, or $R $ is odd and $\eta(\beta)=\eta(-1)^{\frac{R-1}{2}}\varepsilon_f$, then
$$
d_{r}(C_{D})=
q^{s-1}-2q^{s-r-1}.
$$
\item[(2.2)] if $q>3$, $R$ is odd and $\eta(\beta)=-\eta(-1)^{\frac{R-1}{2}}\varepsilon_f$, then
$$
d_{r}(C_{D})=
q^{s-1}-2q^{s-r-1}.
$$
\item[(3)] When $ r= s$, we have
$$
d_{r}(C_{D})=q^{s-1}.
$$
\end{itemize}
\end{theorem}
\begin{proof}
\rm{(1)} We discuss case by case. \\
\textbf{Case 1:}\ $R $ is even and $\varepsilon_{f}=\eta(-1)^{\frac{R}{2}}$. In this case, that is, $1\leq r \leq l_0$ with $l_0 = \frac{R}{2}$.
Suppose $H_r$ is an $r$-dimensional subspace of $\F$ and $\alpha\in  \Prj_{2}(H_{r})$.
By Lemma~\ref{lem:d_r:2}, we have
\begin{align}
N(H_r)=q^{s-(r+1)}\big[1+q^{-\frac{R}{2}}\sum\limits_{(u,-\alpha)\in H_{r}} \upsilon(\beta+f(x_u))\big].   \nonumber
\end{align}
 By Lemma~\ref{lem:2-2}, there exists an $r$-dimensional subspace $J_{r}$ of $\F_{q_1}$ such that $R_{J_{r}}=1, J_r\cap \F_{q_1}^{\perp_f}=\{0\}$ and $f(\alpha_{0})=-4\beta$, for some $\alpha_{0}\in J_{r}$, which concludes that there exists an $(r-1)$-dimensional subspace $J_{r-1}$ of $J_{r}$ satisfying $f(J_{r-1})=0$.
Let $\alpha_{1},\alpha_{2},\cdots,\alpha_{r-1}$ be an $\mathbb{F}_{q}$-basis of $J_{r-1}$. 
Set
$$
\mu_{1}=\alpha_{1}+\alpha_{0},\mu_{2}=\alpha_{2}+\alpha_{0},\cdots,\mu_{r-1}=\alpha_{r-1}+\alpha_{0},\mu_{r}=\alpha_{0},
$$
it's obvious that $\mu_{1},\mu_{2},\cdots,\mu_{r-1},\mu_{r}$ is an $\mathbb{F}_{q}$-basis of $J_{r}$.
Taking
$
H_{r}=\Big\langle(L_f(\mu_{1}),-\alpha),(L_f(\mu_{2}),-\alpha),\cdots,(L_f(\mu_{r-1}),-\alpha),(L_f(\mu_{r}),-\alpha)\Big\rangle
$,
then $N(H_r)$ reaches its maximum
$$ N(H_r)=q^{s-(r+1)}\Big(1+(q-1)q^{r-1-\frac{R}{2}}\Big) = q^{s-(r+1)} + (q-1)q^{s-(2+\frac{R}{2})}. $$
So, the desired result is obtained by Lemma~\ref{lem:d_r:2} and \eqref{eq:d_r:3}.

\textbf{Case 2:}\ $R $ is even and $\varepsilon_{f}=-\eta(-1)^{\frac{R}{2}}$. In this case, $l_0 = \frac{R}{2}$, that is, $1\leq r\leq\frac{R}{2}$.
Suppose $H_r$ is an $r$-dimensional subspace of $\F$ and $\alpha\in  \Prj_{2}(H_{r})$.
By Lemma~\ref{lem:d_r:2}, we have
\begin{align}
N(H_r)=q^{s-(r+1)}\big[1-q^{-\frac{R}{2}}\sum\limits_{(u,-\alpha)\in H_{r}} \upsilon(\beta+f(x_u))\big].   \nonumber
\end{align}
By Lemma~\ref{lem:2-2}, there exists an $r$-dimensional subspace $J_{r}$ of $\F_{q_1}$ such that $R_{J_{r}}=1, J_r\cap \F_{q_1}^{\perp_f}=\{0\}$ and $f(\alpha_{0})=-4\beta+4$, , for some $\alpha_{0}\in J_{r}$, which concludes that there exists an $(r-1)$-dimensional subspace $J_{r-1}$ of $J_{r}$ satisfying $f(J_{r-1})=0$.
Let $\alpha_{1},\alpha_{2},\cdots,\alpha_{r-1}$ be an $\mathbb{F}_{q}$-basis of $J_{r-1}$, then
$$
\mu_{1}=\alpha_{1}+\alpha_{0},\mu_{2}=\alpha_{2}+\alpha_{0},\cdots,\mu_{r-1}=\alpha_{r-1}+\alpha_{0},\mu_{r}=\alpha_{0},
$$
 is an $\mathbb{F}_{q}$-basis of $J_{r}$. 
Taking
$$
H_{r}=\Big\langle(L_f(\mu_{1}),-\alpha),(L_f(\mu_{2}),-\alpha),\cdots,(L_f(\mu_{r-1}),-\alpha),(L_f(\mu_{r}),-\alpha)\Big\rangle,
$$
then $N(H_r)$ reaches its maximum
$$ N(H_r)=q^{s-(r+1)}\Big(1+q^{r-1-\frac{R}{2}}\Big) = q^{s-(r+1)} + q^{s-(2+\frac{R}{2})}. $$
So, the desired result is obtained by Lemma~\ref{lem:d_r:2} and \eqref{eq:d_r:3}.

\textbf{Case 3:}\ \ $R $ is odd. In this case, that is, $1\leq r\leq l_0$ with $l_0 = \frac{R-1}{2}$. Suppose $H_r$ is an $r$-dimensional subspace of $\F$ and $\alpha\in  \Prj_{2}(H_{r})$.
By Lemma~\ref{lem:d_r:2}, we have
\begin{align}\label{eq:nh}
|N(H_r)|=
    	q^{s-(r+1)}\big[1+\varepsilon_f(p^*)^{-\frac{m(R-1)}{2}}\sum\limits_{(u,-\alpha)\in H_{r}}\eta(\beta+f(x_u))\big].
\end{align}
Take an element $a\in\F_q^*$ with $\eta(\beta+a) =\eta(-1)^{\frac{R-1}{2}}\varepsilon_f$. By Lemma~\ref{lem:2-2}, there exists an $r$-dimensional subspace $J_{r}$ of $\F_{q_1}$ such that $R_{J_{r}}=1, J_r\cap \F_{q_1}^{\perp_f}=\{0\}$ and $f(\alpha_{0})=4a$, for some $\alpha_{0}\in J_{r}$, which concludes that there exists an $(r-1)$-dimensional subspace $J_{r-1}$ of $J_{r}$ satisfying $f(J_{r-1})=0$.
Let $\alpha_{1},\alpha_{2},\cdots,\alpha_{r-1}$ be an $\mathbb{F}_{q}$-basis of $J_{r-1}$, then
$$
\mu_{1}=\alpha_{1}+\alpha_{0},\mu_{2}=\alpha_{2}+\alpha_{0},\cdots,\mu_{r-1}=\alpha_{r-1}+\alpha_{0},\mu_{r}=\alpha_{0},
$$
 is an $\mathbb{F}_{q}$-basis of $J_{r}$. 
Set
$$
H_{r}=\Big\langle(L_f(\mu_{1}),-\alpha),(L_f(\mu_{2}),-\alpha),\cdots,(L_f(\mu_{r-1}),-\alpha),(L_f(\mu_{r}),-\alpha)\Big\rangle.
$$
Then $N(H_r)$ reaches its maximum
$$ N(H_r)=q^{s-(r+1)}\Big(1+q^{r-1-\frac{R-1}{2}}\Big) = q^{s-(r+1)} + q^{s-(2+\frac{R-1}{2})}. $$
So, the desired result is obtained by Lemma~\ref{lem:d_r:2} and \eqref{eq:d_r:3}.

\rm{(2)} We also discuss case by case. \\
\textbf{Case 1:} $R $ is even, or $R $ is odd and $\eta(\beta)=\eta(-1)^{\frac{R-1}{2}}\varepsilon_f$, when $l_0< r\leq s-1$, we have 
$1\leq s-r\leq s-l_0-1$. Let $T_{\alpha}=\Big\{x\in \mathbb{F}_{q_2}:\ \mathrm{Tr}_q^{q_2}(\alpha x)=0\Big\}$.
It is easy to know that $\dim(T_{\alpha})=s_2-1$.
By the formula~\eqref{eq:d_r:2} and \cite[Theorem 2]{LL20-0} and Lemma~\ref{lem:2-2}, we can obtain that $\max\{|H_{s_1-l_0}\cap \overline{D}_{\beta}|:H_{s_1-l_0}\in [\E,s_1-l_0]_q\}=2q^{s_1-l_0-1}$, which concludes that there exists an $(s_1-l_0)$-dimensional subspace $J$ of $\E$ such that $|J\cap \overline{D}_{\beta}|=2q^{s_1-l_0-1}$, that is, $R_J=1, \varepsilon_{J}=\eta(\beta)$.
Hence $H_{s-r} = J_1\times T_\alpha$ is a $(s-r)$-dimensional subspace of $J\times T_{\alpha}$ such that $\dim(J_1)\geq 1$ and $R_{J_1}=1$, then we have $|D\cap H_{s-r}| =\max\{|D\cap H_{s-r}|:H_{s-r}\in[\F,s-r]_q\}=2q^{s-r-1}$,
then, $$d_{r}(C_D)=q^{s-1}-2q^{s-r-1}.$$ 

\textbf{Case 2:} $q>3, R $ is odd and $\eta(\beta)=-\eta(-1)^{\frac{R-1}{2}}\varepsilon_f$, in this case $l_0=\frac{R-1}{2}$. 

When $r=l_0+1$, by the proof of \cite[Theorem 1]{LF21}, there exists an $(l_0+1)$-dimensional subspace $J$ of $\E$ such that $R_J=1, \varepsilon_{J}=\eta(-1)^{\frac{R-1}{2}}\varepsilon_f, J\cap \F_{q_1}^{\perp_f}=\{0\}$, which concludes that $|J\cap \overline{D}_c| =2q^{l_0}$ for any $c\in\F_q$ with $\eta(c)=\eta(-1)^{\frac{R-1}{2}}\varepsilon_f$. On the other hand, by $I_2(1)=\sum\limits_{x\in\F_q}\eta(1+x^2)=-1$, we can conclude that there exists an element $b\in\F_q^*$ such that $\eta(1+\gamma b^2)=-1$ for some non-square element $\gamma\in\F_q^*$, thus 
$\eta(\beta+\gamma\beta b^2)= \eta(-1)^{\frac{R-1}{2}}\varepsilon_f$.
Let $c=\gamma\beta b^2$, then by the above analysis, $|J\cap\overline{D}_c| =2q^{l_0}$. Take $\alpha_{0}\in J$ satisfying $f(\alpha_{0})=4c$. From $R_J=1$, there exists an $l_0$-dimensional subspace $J'$ of $J$ satisfying $f(J')=0$.
Let $\alpha_{1},\alpha_{2},\cdots,\alpha_{l_0}$ be an $\mathbb{F}_{q}$-basis of $J'$, then
$$
\mu_{1}=\alpha_{1}+\alpha_{0},\mu_{2}=\alpha_{2}+\alpha_{0},\cdots,\mu_{r-1}=\alpha_{l_0}+\alpha_{0},\mu_{l_0+1}=\alpha_{0},
$$
 is an $\mathbb{F}_{q}$-basis of $J$. 
Set
$$
H_{l_0}=\Big\langle(L_f(\mu_{1}),-\alpha),(L_f(\mu_{2}),-\alpha),\cdots,(L_f(\mu_{r-1}),-\alpha),(L_f(\mu_{r}),-\alpha)\Big\rangle.
$$
By Lemma~\ref{lem:d_r:2} and \eqref{eq:nh}, $N(H_{l_0+1})$ reaches its maximum $|N(H_r)|=2q^{s-1-(l_0+1)}$, hence $$d_{l_0+1}=q^{s-1}-2q^{s-1-(l_0+1)}.$$

 When $l_0+1< r\leq s-1$, that is, $1\leq s- r\leq s_1-R+l_0+s_2-1$. By Lemma~\ref{lem:2-2}, there exists an $(s_1-R+l_0)$-dimensional subspace $J$ of $\E$ such that $R_J=1, \varepsilon_{J}=\eta(\beta)$.
Note that the dimension of the subspace $J\times T_{\alpha}$ is $s-(l_0+2)$. Let $H_{s-r} = J_3\times H_2$ be a $(s-r)$-dimensional subspace of $J\times T_{\alpha}$ such that $\dim(J_3)\geq 1$ and $R_{J_3}=1$,
 then we have $|D\cap H_{s-r}| =\max\{|D\cap H_{s-r}|:H_{s-r}\in[\F,s-r]_q\}=2q^{s-r-1}$,
then, $$d_{r}(C_D)=q^{s-1}-2q^{s-r-1}.$$ 

\rm{(3)} When $ r= s$, by Lemma~\ref{lem:length} we have
$$
d_{r}(C_{\mathrm{D}})=n-\max\Big\{|D \cap H|: H \in [\F,s-r]_{q}\Big\}=n=q^{s-1}.
$$
\end{proof}

\section{Concluding Remarks}
 In this paper, using quadratic forms and trace function, we obtained a class of three-weight linear codes by a bivariate construction. By exponential sum theory and a formula of generalized Hamming weight, we determined the complete
 weight enumerators and weight hierarchies of these linear codes. Our results extended the recent work in \cite{LL22} and \cite{HLL24}. Most of the codes we constructed
are minimal, so they are suitable for constructing secret sharing schemes with interesting properties.

\section*{Acknowledgements}
 The research is supported by the National Science Foundation of China Grant No.12001312,
Key Projects in Natural Science Research of Anhui Provincial Department of Education No.2022AH050594 and Anhui Provincial Natural Science Foundation No.1908085MA02.

\section*{Declaration of competing interest}

The authors declare that they have no known competing financial interests or personal relationships that could have appeared to influence the work
reported in this paper.

\end{document}